\documentclass[11pt]{llncs} 
\pdfoutput=1

\usepackage[utf8]{inputenc} 
\usepackage[T1]{fontenc}

\usepackage{graphicx} 

\usepackage{url}

\newcommand{\inSAT}{\textsc{MAX-1-in-3-SAT}} 

\newcommand{\E}{\mathbf{E}}

\begin{document}

\title{Improved Inapproximability for TSP}
\author{Michael Lampis \thanks{Research supported by ERC Grant 226203}}
\institute{KTH Royal Institue of Technology \\ \email{mlampis@kth.se}}

\maketitle

\begin{abstract}

The Traveling Salesman Problem is one of the most studied problems in
computational complexity and its approximability has been a long standing open
question.  Currently, the best known inapproximability threshold known is
$\frac{220}{219}$ due to Papadimitriou and Vempala. Here, using an essentially
different construction and also relying on the work of Berman and Karpinski on
bounded occurrence CSPs, we give an alternative and simpler inapproximability
proof which improves the bound to $\frac{185}{184}$.

 \end{abstract}

\section{Introduction}

The Traveling Salesman Problem (TSP) is one of the most widely studied
algorithmic problems and deriving optimal approximability results for it has
been a long-standing question. Recently, there has been much progress in the
algorithmic front, after more than thirty years, at least in the important
special case where the instance metric is derived from an unweighted graph,
often referred to as Graphic TSP.  The $\frac{3}{2}$-approximation algorithm by
Christofides was the best known until Gharan et al. gave a slight improvement
\cite{GharanSS11} for Graphic TSP. Then an algorithm with approximation ratio
$1.461$ was given by M{\"o}mke and Svensson \cite{MomkeS11}. With improved
analysis on their algorithm Mucha obtained a ratio of $\frac{13}{9}$
\cite{Mucha12}, while the best currently known algorithm has ratio $1.4$ and is
due to Seb{\"o} and Vygen \cite{Vygen}.

Nevertheless, there is still a huge gap between the guarantee of the best
approximation algorithms we know and the best inapproximability results. The
TSP was first shown MAXSNP-hard in \cite{pap93}, where no explicit
inapproximability constant was derived. The work of Engerbretsen
\cite{Engebretsen03} and B{\"o}ckenhauer et al. \cite{BockenhauerHKSU00} gave
inapproximability thresholds of $\frac{5381}{5380}$ and $\frac{3813}{3812}$
respectively. Later, this was improved to $\frac{220}{219}$ in \cite{pap06} by
Papadimitriou and Vempala\footnote{The reduction of \cite{pap06} was first
presented in \cite{pap00}, which (erroneously) claimed a better bound.}.  No
further progress has been made on the inapproximability threshold of this
problem in the more than ten years since \cite{pap00}. 

\textbf{Overview}: Our main objective in this paper is to give a different,
less complicated inapproximability proof for TSP than the one given in
\cite{pap00,pap06}.  The proof of \cite{pap06} is very much optimized to
achieve a good constant: the authors reduce directly from MAX-E3-LIN2, a
constraint satisfaction problem (CSP) for which optimal inapproximability
results are known, due to Håstad \cite{hastad}. They take care to avoid
introducing extra gadgets for the variables, using only gadgets that encode the
equations.  Finally they define their own custom expander-like notion on graphs
to ensure consistency between tours and assignments. Then the reduction is
performed in essentially one step.

Here on the other hand we take the opposite approach, choosing simplicity over
optimization. We also start from MAX-E3-LIN2 but go through two intermediate
CSPs. The first step in our reduction gives a set of equations where each
variable appears at most five times (this property will come in handy in the
end when proving consistency between tours and assignments). In this step,
rather than introducing something new we rely heavily on machinery developed by
Berman and Karpinski to prove inapproximability for bounded occurrence CSPs
\cite{BK99,BK01,BK03}. As a second step we reduce to \inSAT. The motivation is
that the \textsc{1-in-3} predicate nicely corresponds to the objectives of TSP,
since we represent clauses by gadgets and the most economical solution will
visit all gadgets once but not more than once. Another way to view this step is
that we use \inSAT\ as an aid to design a TSP gadget for parity.  Finally, we
give a reduction from \inSAT\ to TSP.

  This approach is (at least arguably) simpler than the approach of
\cite{pap06}, since some of our arguments can be broken down into independent
pieces, arguing about the inapproximability of intermediate, specially
constructed CSPs. We also benefit from re-using out-of-the box the amplifier
construction of \cite{BK03}.  Interestingly, putting everything together we end
up obtaining a slightly better constant than the one currently known, implying
that there may still be some room for further improvement. Though we are still
a long way from an optimal inapproximability result, our results show that
there may still be hope for better bounds with existing tools.  Exploring how
far these techniques can take us with respect to TSP (and also its variants,
see for example \cite{karpinski}) may thus be an interesting question.  

The main result of this paper is given below and it follows directly from the
construction in section \ref{sec:construction} and Lemmata
\ref{lem:easy},\ref{lem:hard}.

\begin{theorem}

For all $\epsilon>0$ there is no polynomial-time
$(\frac{92.3}{91.8}-\epsilon)$-approximation algorithm for TSP, unless P=NP.

\end{theorem}

\section{Preliminaries}

We will denote graphs by $G(V,E)$. All graphs are assumed to be undirected,
loop-less and edge-weighted, meaning that there is also a function $w: E \to
\mathbf{R}^+$.  In some cases we will allow $E$ to be a multi-set, that is, we
may allow parallel edges. In the case of a multi-set $E$ that contains several
copies of some elements, when we write $\sum_{e\in E} w(e)$ we mean the sum
that has one term for each copy. A (multi-)graph is Eulerian if there exists a
closed walk that visits all its vertices and uses each edge once. It is well
known that a (multi-)graph is Eulerian iff it is connected and all its vertices
have even degree. We will use $[n]$ to denote the set $\{1,2,\ldots,n\}$. We
will use $\E[X]$ to denote the expectation of a random variable $X$.

In the metric Traveling Salesman Problem (TSP) we are given as input an
edge-weighted undirected graph $G(V,E)$. Let $d(u,v)$, for $u,v\in V$ denote
the shortest-path distance from $u,v$. The objective is to find an ordering
$v_1,v_2,\ldots,v_n$ of the vertices such that $\sum_{i=1}^{n-1}d(v_i,v_{i+1})
+ d(v_n,v_1)$ is minimized.

Another, equivalent view of the TSP is the following: given an edge-weighted
graph $G(V,E)$ we seek to find a multi-set $E_T$ consisting of edges from $E$
such that the graph induced by $E_T$ spans $V$, is Eulerian and the sum of the
weights of all edges in $E_T$ is minimized. It is not hard to see that the two
formulations are equivalent. We will make use of this second formulation
because it makes some arguments on our construction easier.

We generalize the Eulerian multi-graph formulation as follows: a multi-set
$E_T$ of edges from $E$ is a quasi-tour iff the degrees of all vertices in the
multi-graph $G_T(V,E_T)$ are even. The cost of a quasi-tour is defined as
$\sum_{e\in E_T} w(e) + 2(c(G_T)-1)$, where $c(G_T)$ denotes the number of
connected components of the multi-graph. It is not hard to see that a TSP tour
can also be considered a quasi-tour with the same cost (since for a normal tour
$c(G_T)=1$), but in a weighted graph there could potentially be a quasi-tour
that is cheaper than the optimal tour.

\subsection{Forced edges}

As mentioned, we will view TSP as the problem of selecting edges from $E$ to
form a minimum-weight multi-set $E_T$ that makes the graph Eulerian. It is easy
to see that no edge will be selected more than twice, since if an edge is
selected three times we can remove two copies of it from $E_T$ and the graph
will still be Eulerian while we have improved the cost. 

In our construction we would like to be able to stipulate that some edges are
to be used at least once in any valid tour. We can achieve this with the
following trick: suppose that there is an edge $(u,v)$ with weight $w$ that we
want to force into every tour. We sub-divide this edge a large number of times,
say $p-1$, that is, we remove the edge and replace it with a path of $p$ edges
going through new vertices of degree two.  We then redistribute the original
edge's weight to the $p$ newly formed edges, so that each has weight $w/p$.
Now, any tour that fails to use two or more of the newly formed edges must be
disconnected.  Any tour that fails to use exactly one of them can be augmented
by adding two copies of the unused edge.  This only increases the cost by
$2w/p$, which can be made arbitrarily small by giving $p$ an appropriately
large value.  Therefore, we may assume without loss of generality that in our
construction we can force some edges to be used at least once. Note that these
arguments apply also to quasi-tours.

\section{Intermediate CSPs}

In this section we will design and prove inapproximability for a family of
instances of \inSAT\ with some special structure. We will use these instances
(and their structure) in the next section where we reduce from \inSAT\ to TSP.

Let $I_1$ be a system of $m$ linear equations mod 2, each consisting of exactly
three variables. Let $n$ be the total number of variables appearing in $I_1$
and let the variables be denoted as $x_i, i\in [n]$.  Let $B$ be the maximum
number of times any variable appears. We will make use of the following seminal
result due to Håstad:

\begin{theorem}[\cite{hastad}]

For all $\epsilon>0$ there exists a $B$ such that given an instance $I_1$ as
above it is NP-hard to decide if there is an assignment that satisfies at least
$(1-\epsilon)m$ equations or all assignment satisfy at most
$(\frac{1}{2}+\epsilon)m$ equations.

\end{theorem}

\subsection{Bounded Occurences}

In $I_1$ each variable appears at most a constant number of times $B$, where
$B$ depends on $\epsilon$. We would like to reduce the maximum number of
occurrences of each variable to a small absolute constant. For this, one
typically uses some kind of expander or amplifier construction. Here we will
rely on a construction due to Berman and Karpinski that reduces the number of
occurrences to 5.

\begin{theorem}[\cite{BK03}] \label{thm:BK}

Consider the family of bipartite graphs $G(L,R,E)$, where $|L|=B, |R|=0.8B$,
all vertices of $L$ have degree 4,  all vertices of $R$ have degree 5 and $B$
is a sufficiently large multiple of 5.  If we select uniformly at random a
graph from this family then with high probability it has the following
property: for any $S\subseteq L\cup R$ such that $|S\cap L|\le \frac{|L|}{2}$
the number of edges in $E$ with exactly one endpoint in $S$ is at least $|S\cap
L|$.

\end{theorem}

We now use the above construction to construct a system of equations where each
variable appears exactly 5 times. First, we may assume that in $I_1$ the number
of appearances of each variable is a multiple of 5 (otherwise, repeat all
equations five times). Also, by repeating all the equations we can make sure
that all variables appear at least $B'$ times, where $B'$ is a sufficiently
large number to make Theorem \ref{thm:BK} hold. 

For each variable $x_i$ in $I_1$ we introduce the variables $x_{(i,j)},
j\in[d(i)]$ and $y_{(i,j)}, j\in[0.8d(i)]$ where $d(i)$ is the number of
appearances of $x_i$ in the original instance. We call $X_i = \{ x_{(i,j)}\ |\
j\in[d(i)] \} \cup \{y_{(i,j)}\ |\  j\in[0.8d(i)] \}$ the cloud that
corresponds to $x_i$. Construct a bipartite graph with the property described
in Theorem \ref{thm:BK} with $L=[d(i)], R=[0.8d(i)]$ (since $d(i)<B$ is a
constant that depends only on $\epsilon$ this can be done in constant time by
brute force).  For each edge $(j,k)\in E$ introduce the equation $x_{(i,j)} +
y_{(i,k)} = 1$. Finally, for each equation $x_{i_1} + x_{i_2} + x_{i_3} = b$ in
$I_1$, where this is the $j_1$-th appearance of $x_{i_1}$, the $j_2$-th
appearance of $x_{i_2}$ and the $j_3$-th appearance of $x_{i_3}$ replace it
with the equation $x_{(i_1,j_1)} + x_{(i_2,j_2)} + x_{(i_3,j_3)} = b$.

Denote this instance by $I_2$ and we have $|I_2|=13m$, with $12m$ equations
having size 2.  A consistent assignment to a cloud $X_i$ is an assignment that
sets all $x_{(i,j)}$ to $b$ and all $y_{(i,j)}$ to $1-b$.  By standard
arguments using the graph of Theorem \ref{thm:BK} we can show that an optimal
assignment to $I_2$ is consistent (in each inconsistent cloud let $S$ be the
vertices with the minority assignment; flipping all variables of $S$ cannot
make the solution worse). From this it follows that it is NP-hard to
distinguish if the maximum number of satisfiable equations is at least
$(13-\epsilon)m$ or at most $(12.5+\epsilon)m$.

\subsection{\inSAT}

In the \inSAT\ problem we are given a collection of clauses $(l_i \lor l_j \lor
l_k)$, each consisting of at most three literals, where each literal is either
a variable or its negation.  A clause is satisfied by a truth assighment if
exactly one of its literals is set to True.  The problem is to find an
assignment that satisfies the maximum number of clauses.

We would like to produce a \inSAT\ instance from $I_2$. Observe that it is easy
to turn the size two equations $x_{(i,j)} + y_{(i,k)} =1$ to the equivalent
clauses $(x_{(i,j)} \lor y_{(i,k)})$. We only need to worry about the $m$
equations of size three.

If the $k$-th size-three equation of $I_2$ is $x_{(i_1,j_1)} + x_{(i_2,j_2)} +
x_{(i_3,j_3)} = 1$ we introduce three new auxilliary variables $a_{(k,i)},
i\in[3]$ and replace the equation with the three clauses $(x_{(i_1,j_1)}\lor
a_{(k,1)}\lor a_{(k,2)})$, $(x_{(i_2,j_2)}\lor a_{(k,2)}\lor a_{(k,3)})$,
$(x_{(i_3,j_3)}\lor a_{(k,1)}\lor a_{(k,3)})$. If the right-hand-side of the
equation is $0$ then we add the same three clauses except we negate
$x_{(i_1,j_1)}$ in the first clause. We call these three clauses the cluster
that corresponds to the $k$-th equation. 

It is not hard to see that if we fix an assignment to $x_{(i_1,j_1)},
x_{(i_2,j_2)}, x_{(i_3,j_3)}$ that satisfies the $k$-th equation of $I_2$ then
there exists an assignment to $a_{(k,1)}, a_{(k,2)}, a_{(k,3)}$ that satisfies
the whole cluster. Otherwise, at most two of the clauses of the cluster can be
satisfied. Furthermore, in this case there exist three different assignments to
the auxilliary variables that satisfy two clauses and each leaves a different
clause unsatisfied.

  From now on, we will denote by $M$ the set of (main) variables $x_{(i,j)}$,
by $C$ the set of (checker) variables $y_{(i,j)}$ and by $A$ the set of
(auxilliary) variables $a_{(k,i)}$. Call the instance of \inSAT\ we have
constructed $I_3$. Note that it consists of $15m$ clauses and $8.4m$ variables.

\section{TSP}

\subsection{Construction} \label{sec:construction}

\begin{figure}

\centering

\includegraphics[scale=0.2]{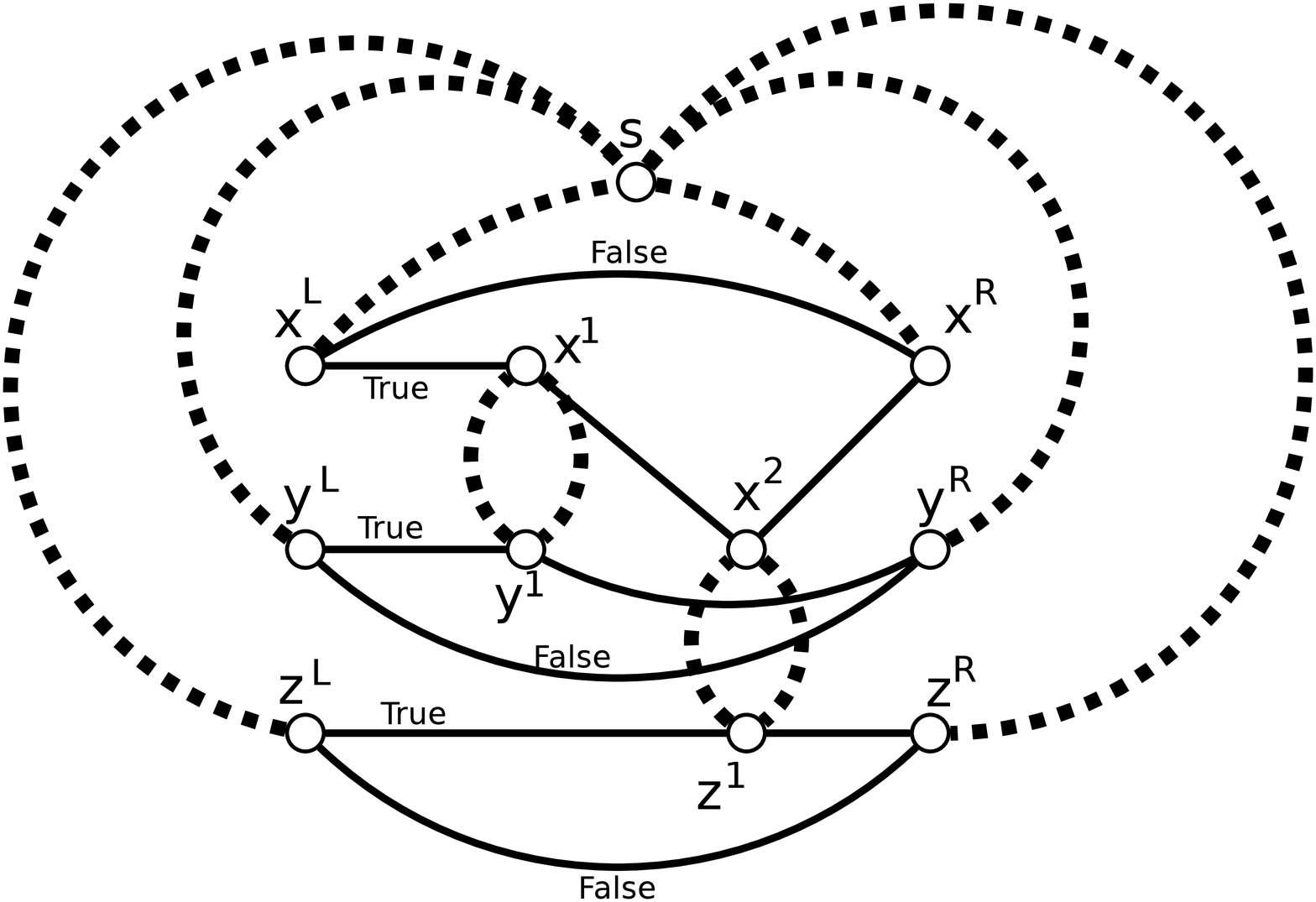}

\caption{Example construction for the clause $(x\lor y) \land (x\lor z)$.
Forced edges are denoted by dashed lines. There are two terminals for each
variable and two gadgets that represent the two clauses. The True edges
incident on the terminals are re-routed through the gadgets where each variable
appears positive. The False edges connect the terminals directly since no
variable appears anywhere negated.} \label{fig:basic}

\end{figure}

We now describe a construction that encodes $I_3$ into a TSP instance $G(V,E)$.
Rather than viewing this as a generic construction from \inSAT\ to TSP, we will
at times need to use facts that stem from the special structure of $I_3$. In
particular, the fact that variables can be partitioned into sets $M,C,A$, such
that variables in $M\cup C$ appear five times and variables in $A$ appear
twice; the fact that most clauses have size two and they involve one positive
variable from $M$ and one positive variable from $C$; and also the fact that
clauses of size three come in clusters as described in the construction of
$I_3$. 

As mentioned, we assume that in the graph $G(V,E)$ we may include some forced
edges, that is, edges that have to be used at least once in any tour.  The
graph includes a central vertex, which we will call $s$. For each variable in
$x\in M\cup C\cup A$ we introduce two new vertices named $x^L$ and $x^R$, which
we will call the left and right terminal associated with $x$. We add a forced
edge from each terminal to $s$. For terminals that correspond to variables in
$M\cup C$ this edge has weight $7/4$, while for variables in $A$ it has weight
$1/2$.  We also add two (parallel) non-forced edges between each pair of
terminals representing the same variable, each having a weight of 1 (we will
later break down at least one from each pair of these, so the graph we will
obtain in the end will be simple).  Informally, these two edges encode an
assignment to each variable: we arbitrarily label one the True edge and the
other the False edge, the idea being that a tour should pick exactly one of
these for each variable and that will give us an assignment. We will re-route
these edges through the clause gadgets as we introduce them, depending on
whether each variable appears in a clause positive or negative.

Now, we add some gadgets to encode the size-two clauses of $I_3$. Let
$(x_{(i,j_1)} \lor y_{(i,j_2)} )$ be a clause of $I_3$ and suppose that this is
the $k_1$-th clause that contains $x_{(i,j_1)}$ and the $k_2$-th clause that
contains $y_{(i,j_2)}$, $k_1,k_2\in[5]$. Then we add two new vertices to the
graph, call them $x_{(i,j_1)}^{k_1}$ and $y_{(i,j_2)}^{k_2}$. Add two forced
edges between them, each of weight $3/2$ (recall that forced edges represent
long paths, so these are not really parallel edges). Finally, re-route the True
edges incident on $x_{(i,j_1)}^L$ and $y_{(i,j_2)}^L$ through
$x_{(i,j_1)}^{k_1}$ and $y_{(i,j_2)}^{k_2}$ respectively. More precisely, if
the True edge incident on $x_{(i,j_1)}^L$ connects it to some other vertex $u$,
remove that edge from the graph and add an edge from $x_{(i,j_1)}^L$ to
$x_{(i,j_1)}^{k_1}$ and an edge from $x_{(i,j_1)}^{k_1}$ to $u$. All these
edges have weight one and are non-forced (see Figure \ref{fig:basic}).

We use a similar gadget for clauses of size three. Consider a cluster
$(x_{(i_1,j_1)}\lor a_{(k,1)}\lor a_{(k,2)})$, $(x_{(i_2,j_2)}\lor
a_{(k,2)}\lor a_{(k,3)})$, $(x_{(i_3,j_3)}\lor a_{(k,1)}\lor a_{(k,3)})$ and
suppose for simplicity that this is the fifth appearance for all the main
variables of the cluster. Then we add the new vertices $x_{(i_1,j_1)}^5,
x_{(i_2,j_2)}^5, x_{(i_3,j_3)}^5$ and also the vertices $a_{(k,1)}^{1},
a_{(k,1)}^{2}, a_{(k,2)}^{1}, a_{(k,2)}^{2}$ and $a_{(k,3)}^{1}, a_{(k,3)}^{2}
$. To encode the first clause we add two forced edges of weight $5/4$, one from
$x_{(i_1,j_1)}^5$ to $a_{(k,1)}^{1}$ and one from $x_{(i_1,j_1)}^5$ to
$a_{(k,2)}^{1}$. We also add a forced edge of weight 1 from   $a_{(k,1)}^{1}$
to $a_{(k,2)}^{1}$, thus making a triangle with the forced edges (see Figure
\ref{fig:basic2}). We re-route the True edge from $a_{(k,1)}^L$ through
$a_{(k,1)}^{1}$ and $a_{(k,1)}^{2}$. We do similarly for the other two
auxilliary variables and the main variables. Finally, for a cluster where
$x_{(i_1,j_1)}$ is negated, we use the same construction except that rather
than re-routing the True edge that is incident on $x_{(i_1,j_1)}^L$ we re-route
the False edge.  This completes the construction.

\begin{figure}

\centering

\includegraphics[scale=0.2]{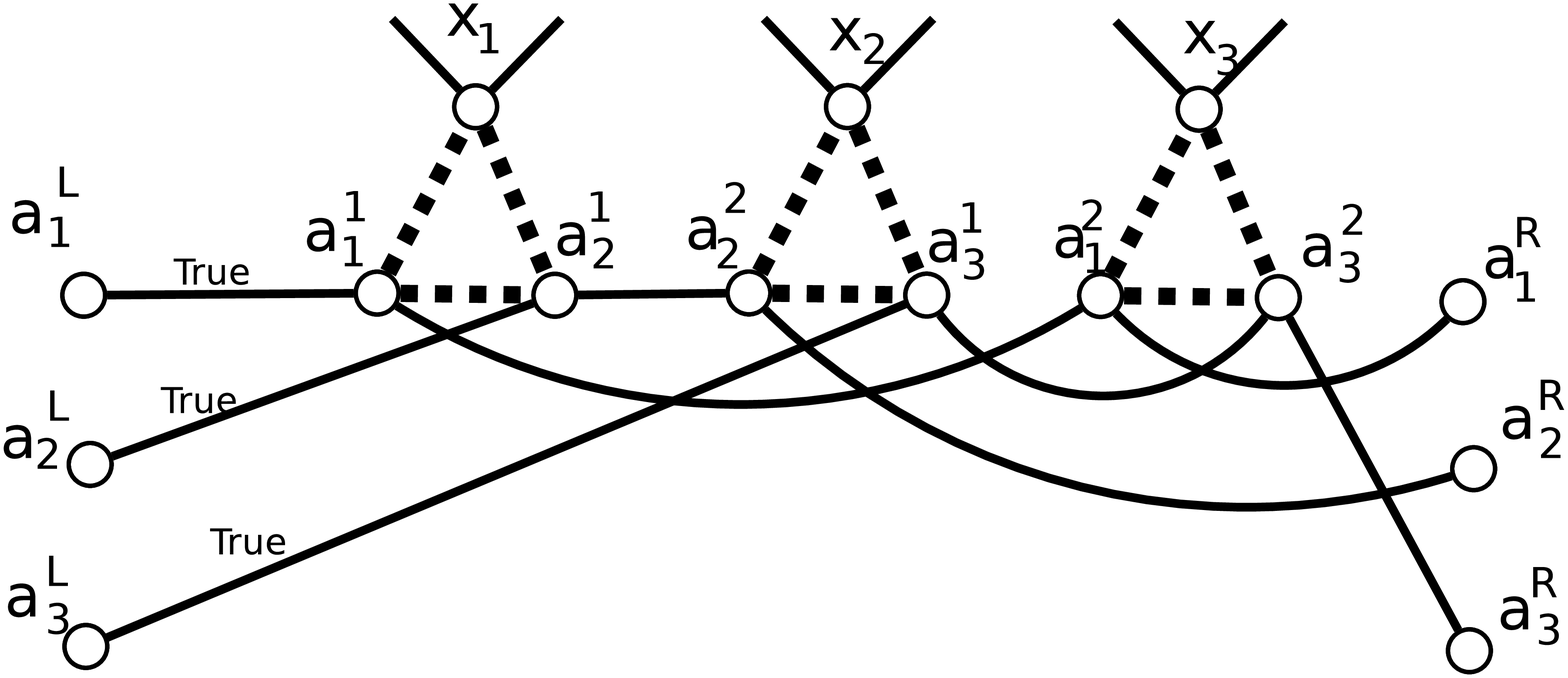}

\caption{Example construction fragment for the cluster $(x_1\lor a_1 \lor a_2)
\land (x_2 \lor a_2 \lor a_3) \land (x_3\lor a_1\lor a_3)$. The False edges
which connect each pair of terminals and the forced edges that connect
terminals to $s$ are not shown.} \label{fig:basic2}

\end{figure}

\subsection{From Assignment to Tour}

Let us now prove one direction of the reduction and in the process also give
some intuition about the construction. Call the graph we have constructed
$G(V,E)$.

\begin{lemma} \label{lem:easy}

If there exists an assignment to the variables of $I_3$ that leaves at most $k$
equations unsatisfied, then there is a tour of $G$ with cost at most $T=L+k$,
where $L=91.8m$.

\end{lemma}

\begin{proof}

Observe that by construction we may assume that all the unsatisfied clauses of
$I_3$ are in the clusters and that at most one clause in each cluster is
unsatisfied, otherwise we can obtain a better assignment. Also, if an
unsatisfied clause has all literals set to False we can flip the value of one
of the auxilliary variables without increasing the number of violated clauses.
Thus, we may assume that all clauses have a True literal. Also, we may assume
that no clause has all literals set to True: suppose that a clause does, then
both auxilliary variables of the clause are True. We set them both to False,
gaining one clause. If this causes the two other clauses of the cluster to
become unsatisfied, set the remaining auxilliary variable to True. We conclude
that all clauses have either one or two True literals.

Our tour uses all forced edges exactly once. For each variable $x$ set to True
in the assignment the tour selects the True edge incident on the terminal
corresponding to $x$. If the edge has been re-routed all its pieces are
selected, so that we have selected edges that make up a path from $x^L$ to
$x^R$.  Otherwise, if $x$ is set to False in the assignment the tour selects
the corresponding False path.

Observe that this is a valid quasi-tour because all vertices have even degree
(for each terminal we have selected the forced edge plus one more edge, for
gadget vertices we have selected the two forced edges and possibly the two
edges through which True or False was re-routed). Also, observe that the tour
must be connected, because each clause contains a True literal, therefore for
each gadget two of its external edges have been selected and they are part of a
path that leads to the terminals.

The cost of the tour is at most $F+N+M+k$, where $F$ is the total cost of all
forced edges in the graph and $N, M$ are the total number of variables and
clauses respectively in $I_3$. To see this, notice that there are $2N$
terminals, and there is one edge incident on each and there are $M$ clause
gadgets, $M-k$ of which have two selected edges incident on them and $k$ of
which have four. Summing up, this gives $2N+2M+2k$, but then each unit-weight
edge has been counted twice, meaning that the non-forced edges have a total
cost of $N+M+k$.

Finally, we have $N=8.4m$, $M=15m$ and $F=3\times 12m + \frac{7}{2}\times 3m +
\frac{7}{2}\times 5.4m + 1\times 3m = 68.4m$, where the terms are respectively
the cost of size-two clause gadgets, the cost of size-three clause gadgets, the
cost of edges connecting terminals to $s$ for the main variables and for the
auxilliary variables. We have $F+N+M = 91.8m$. \qed

\end{proof}

\subsection{From Tour to Assignment}

We would like now to prove the converse of Lemma \ref{lem:easy}, namely that if
a tour of cost $L+k$ exists then we can find an assignment that leaves at most
$k$ clauses unsatisfied. Let us first give some high-level intuition and in the
process justify the weights we have selected in our construction.

Informally, we could start from a simple base case: suppose that we have a tour
such that all edges of $G$ are used at most once. It is not hard to see that
this then corresponds to an assignment, as in the proof of Lemma
\ref{lem:easy}. So, the problem is how to avoid tours that may use some edges
twice. 

To this end, we first give some local improvement arguments that make sure that
the number of problematic edges, which are used twice, is limited.  However,
arguments like these can only take us so far, and we would like to avoid having
too much case analysis.

We therefore try to isolate the problem. For variables in $M\cup C$ which the
tour treats honestly, that is, variables which are not involved with edges used
twice, we directly obtain an assignment from the tour. For the other variables
in $M\cup C$ we pick a random value and then extend the whole assignment to $A$
in an optimal way. We want to show that the expected number of unsatisfied
clauses is at most $k$.

The first point here is that if a clause containing only honest variables turns
out to be violated, the tour must also be paying an extra cost for it. The
difficulty is therefore concentrated on clauses with dishonest variables.

By using some edges twice the tour is paying some cost on top of what is
accounted for in $L$. We would like to show that this extra cost is larger than
the number of clauses violated by the assignment. It is helpful to think here
that it is sufficient to show that the tour pays an additional cost of
$\frac{5}{2}$ for each dishonest variable, since main variables appear 5 times.

A crucial point now is that, by a simple parity argument, there has to be an
even number of violations (that is, edges used twice) for each variable (Lemma
\ref{lem:even}).  This explains the weights we have picked for the forced edges
in size-three gadgets ($\frac{5}{4}$) and for edges connecting terminals to $s$
($\frac{7}{4}=\frac{5}{4}+\frac{1}{2}$ or $\frac{5}{4}$ extra to the cost
already included in $L$ for fixing the parity of the terminal vertex). Two such
violations give enough extra cost to pay for the expected number of unsatisfied
clauses containing the variable.

At this point, we could also set the weights of forced edges in size-two
gadgets to $\frac{5}{2}$, which would be split among the two dishonest
variables giving $\frac{5}{4}$ to each. Then, any two violations would have
enough additional cost to pay for the expected unsatisfied clauses. However, we
are slightly more careful here: rather than setting all dishonest variables in
$M\cup C$ independently at random, we pick a random but consistent assignment
for each cloud. This ensures that all size-two clauses with violations will be
satisfied. Thus, it is sufficient for violations in them to have a cost of
$\frac{3}{2}$: the amount "paid" to each variable is now
$\frac{3}{4}=\frac{5}{4} - \frac{1}{2}$, but the expected number of unsatisfied
clauses with this variable is also decreased by $\frac{1}{2}$ since one clause
is surely satisfied.

\medskip

Let us now proceed to give the full details of the proof.  Recall that if a
tour of a certain cost exists, then there exists also a quasi-tour of the same
cost.  It suffices then to prove the following:

\begin{lemma} \label{lem:hard}

If there exists a quasi-tour of $G$ with cost at most $L+k$ then there exists
an assignment to the variables of $I_3$ that leaves at most $k$ clauses
unsatisfied.

\end{lemma}

In order to prove Lemma \ref{lem:hard} it is helpful to first make some easy
observations. First, observe that if a quasi-tour uses a unit-weight edge twice
then we can remove both of these appearances of the edge from the solution
without increasing the cost, since the number of components can only increase
by one.  Therefore, all (non-forced) edges of weight one are used at most once.

Second, if both forced edges of a gadget of size two are used twice then we can
remove one appearance of each from the solution, decreasing the cost.
Similarly, in a gadget of size three if two forced edges are used twice then we
can drop one copy of each and use the third edge twice, making the tour
cheaper.  Therefore, in each gadget there is at most one forced edge that is
used twice.

Third,  if both forced edges that connect the terminals $x^L, x^R$ to $s$ are
used twice, then we can remove one appearance of each from the solution and
replace them by the shortest path from $x^L$ to $x^R$ that uses only non-forced
unit weight edges.  This has weight at most one for the auxilliary variables
and two for the rest, which in both cases is at most as much as the weight of
the removed edges. Therefore, for each variable $x$, at least one of the forced
edges that connect $x^L, x^R$ to $s$ is used exactly once.

Given a tour $E_T$, we will say that a variable $x$ is honestly traversed in
that tour if all the forced edges that involve it are used exactly once (this
includes the forced edges incident on $x^L, x^R$ and $x^i$, $i\in [5]$).

Let us now give two more useful facts.

\begin{lemma} \label{lem:cheat} 

There exists an optimal tour where all forced edges between two different
vertices that correspond to two variables in $A$ are used exactly once.

\end{lemma}

\begin{proof}

We refer the reader again to Figure \ref{fig:basic2}. Suppose for contradiction
that the edge $(a_1^1, a_2^1)$ is used twice (the other cases are equivalent by
symmetry since all vertices $a_i^j$ are connected to one terminal and one other
such vertex).

First, suppose that at least one of the edges that connect one of these two
endpoints to a terminal is selected, say the edge $(a_1^L, a_1^1)$. Then modify
the solution by removing that edge and a copy of the duplicate forced edge and
adding a copy of $(a_2^L,a_2^1)$, $(s,a_2^L)$ and $(s,a_1^L)$. This does not
increase the cost.

Second, suppose that both $(s,a_1^L)$ and $(s,a_2^L)$ are used twice in the
tour. Then we can modify the tour by dropping one copy of each and a copy of
the duplicate gadget edge and adding $(a_1^L,a_1^1)$ and $(a_2^L,a_2^1)$.

Finally, suppose that none of the previous two cases is true. Thus, neither of
$(a_1^L, a_1^1)$, $(a_2^L,a_2^1)$ is used in the tour. This means that
$(a_1^1,a_1^2)$ and $(a_2^1,a_2^2)$ are both used to ensure that $a_1^1, a_2^1$
have even degree. Also, one of the edges connecting a terminal to $s$ is used
once, say $(s,a_1^L)$. This means that the False edge incident to $a_1^L$ must
be used to make the degree of $a_1^L$ even. Remove the False edge and the edge
$(a_1^1,a_1^2)$ from the tour and add the edges $(a_1^L,a_1^1)$ and
$(a_1^R,a_1^2)$. This reduces to the first case. \qed

\end{proof}

\begin{lemma} \label{lem:even}

In an optimal tour, if a variable is dishonest then it must be dishonest twice.
More precisely, the number of forced edges that involve the variable (either
inside gadgets or connecting terminals to $s$) and are used twice must be even.  

\end{lemma}

\begin{proof}

Consider a variable $x$ and first suppose that neither of the forced edges
connecting $s$ to the terminals is used twice, but there is a single forced
edge in a gadget that is used twice. It follows that the vertex that
corresponds to $x$ in that gadget has an odd number of unit-weight edges
incident to it selected.  The two terminals have a single selected unit-weight
edge incident on them and all other vertices that belong to $x$ have an even
number of incident unit-weight edges selected, since their total degree is
even. Thus, summing the number of selected unit-weight edges incident on all
the vertices that belong to $x$ we get an odd number, which is a contradiction
since we counted each such edge exactly twice. A similar argument applies if
one assumes that one of the forced edges incident on the terminals is used
twice and all other forced edges are used once. \qed

\end{proof}

Observe that it follows from Lemmata \ref{lem:cheat},\ref{lem:even} that if all
the main variables involved in a cluster are honest then the auxilliary
variables of that cluster are also honest. This holds because if the main
variables are honest then by Lemma \ref{lem:cheat} no forced edge inside the
gadgets of the cluster is used twice, so by Lemma \ref{lem:even} and the fact
that at least one of the forced edges incident on the terminals is used once,
the auxilliary variables are honest.

We would like now to be able to extract a good assignment even if a tour is not
honest, thus indirectly proving that honest tours are optimal. 

\begin{proof}[Lemma \ref{lem:hard}]

Consider the following algorithm to extract an assignment from the tour: first,
for each variable in $M\cup C$ that was traversed honestly give it the same
truth-value as in the tour, that is, if the tour selects the True edge incident
on the corresponding terminal, set the variable to True, otherwise to False. To
decide on the value of the dishonest variables from $M\cup C$ produce $n$
random bits $b_i, i\in [n]$ (recall that $n$ is the number of variables of
$I_1$, or the number of clouds in $I_2$). For each $i$ set all dishonest
variables $x_{(i,j)}$ to be equal to $b_i$ and all dishonest $y_{(i,j)}$ to be
equal to $1-b_i$. This ensures that size-two clauses that contain two dishonest
variables are always satisfied, since these clauses are always between two
variables of the same cloud.

Let us also assign the auxilliary variables. If there is an assignment to the
auxilliary variables of a cluster that satisfies all three clauses select it.
Otherwise, select an assignment that violates the clause of a dishonest
variable from $M$, if such a variable exists, and satisfies the other two. If
all main variables are honest, as we have argued the auxilliary variables are
also honest, so pick the corresponding assignment.

We now have a randomized assignment for $I_3$, so let us upper-bound the
expected number of unsatisfied clauses. Let $U$ be a random variable equal to
the set of unsatisfied clauses and let $U=U_1\cup U_2$ where $U_1$ contains all
the unsatisfied clauses that involve only honest variables from $M\cup C$ and
$U_2$ the rest. (Note that $U_1$ is not random.) 

The cost of the quasi-tour we have is $T \le F + N + M + k$. Let $E_G$ be the
set of forced gadget edges that the tour uses twice. Let $E_S$ be the set of
forced edges incident on $s$ that the tour uses twice. Let $E_1$ be the set of
unit-weight edges that the tour uses (recall that each is used once). Let
$U_1'$ be the set of clauses that correspond to gadgets the tour visits at
least twice (meaning they have at least four incident edges selected). Let
$U_1''$ be the set of clauses that correspond to gadgets the tour does not
visit (meaning that each forms its own connected component).

We have $T = \sum_{e\in E_T} w(e) + 2(c(G_T) -1) = F + \sum_{e\in E_1} w(e) +
\sum_{e\in E_G} w(e) + \sum_{e\in E_S} w(e) + 2(c(G_T) - 1)$.

By definition $\sum_{e\in E_1} w(e) = |E_1|$. Let us try to lower-bound this
quantity using arguments similar to the proof of Lemma \ref{lem:easy}. After
the selection of the forced edges there are $2N-|E_S|$ terminals with odd
degree, so each has a selected unit-weight edge incident to it. There are
$|U_1'|$ gadgets with at least four selected incident edges and
$M-|U_1'|-|U_1''|$ gadgets with two selected incident edges. Summing up we get
$2N-|E_S|+2M+2|U_1'|-2|U_1''|$, but each edge is counted twice, so we have
$|E_1|\ge N - \frac{1}{2}|E_S| + M + |U_1'| - |U_1''|$. 

Using this fact we get $T \ge F + N + M + \sum_{e\in E_G} w(e) + \sum_{e\in
E_S} ( w(e)-\frac{1}{2}) + |U_1'| + 2(c(G_T) - 1) - |U_1''|$. 

Now, observe that $|U_1''| \le c(G_T) - 1$, because each element of $U_1''$
forms a component and there is one component that is not an element of $U_1''$
(the one that contains $s$). Thus, $2(c(G_T)-1) - |U_1''| \ge |U_1''|$.
Combining this with the above we get $T \ge F + N + M + \sum_{e\in E_G} w(e) +
\sum_{e\in E_S} ( w(e)-\frac{1}{2}) + |U_1'| +  |U_1''|$. Given the known
upper-bound on the cost of the tour we have that $ k \ge \sum_{e\in E_G} w(e) +
\sum_{e\in E_S} ( w(e)-\frac{1}{2}) + |U_1'| + |U_1''|$.

We now need to argue two facts and we are done. First $|U_1| \le |U_1'| +
|U_1''|$. Recall that $U_1$ is the set of unsatisfied clauses that involve
honest variables. Since the variables are traversed honestly their
corresponding gadgets are either visited at least twice or not at all, so they
are counted in $|U_1'|$ or in $|U_1''|$.

Second, we would like to show that $\E[|U_2|] \le \sum_{e\in E_G} w(e) +
\sum_{e\in E_S} ( w(e)-\frac{1}{2})$. Before we do that, observe that if we
show this then it follows that $\E[|U|] = \E[|U_2|] + |U_1| \le k$, so there
must exist an assignment that leaves no more than $k$ clauses unsatisfied and
we are done.

So, let us try to upper-bound $\E[|U_2|]$, which is the expected number of
unsatisfied clauses that contain a dishonest variable. First, observe that if
there are dishonest auxilliary variables in a cluster by the construction of
the assignment we have ensured that any unsatisfied clause must contain a
dishonest main variable. Therefore, it suffices to count the expected number of
unsatisfied clauses that contain a dishonest main variable.

Let us define a credit $cr(x)$ for each dishonest main variable $x$. If a
forced edge connecting a terminal to $s$ is used twice we give $x$ a credit of
$5/4$ (which is equal to $w(e)-\frac{1}{2}$, since these edges have weight
$\frac{7}{4}$). If a forced edge in a gadget that involves $x$ and another main
variable is used twice we give $x$ a credit of $\frac{3}{4}$ (which is equal to
$w(e)/2$). Finally, if a forced edge in a gadget that involves $x$ and an
auxilliary variable is used twice we give $x$ a credit of $\frac{5}{4}$ (which
is equal to $w(e)$). We define $cr(x)$ to be the sum of credits given to $x$ in
this process.

If $D$ is the set of dishonest main variables then it is not hard to see that
$\sum_{x\in D} cr(x) \le \sum_{e\in E_G} w(e) + \sum_{e\in E_S} (
w(e)-\frac{1}{2})$. All edges are counted once in the sum of credits, except
for those from $E_G$ that involve two main variables, for which each is
credited half the weight.

We will now argue that the expected number of unsatisfied clauses that contain
a variable $x$ is at most $cr(x)$. Recall that clauses containing $x$ and
another dishonest main variable are by construction satisfied, while clauses
made up of $x$ and one honest variable are satisfied with probability $1/2$.
Also, clauses of size 3 that contain $x$ are satisfied with probability at
least $1/2$, since with probability $1/2$ the equation from which the cluster
was obtained is satisfied.  Thus, if $cr(x)\ge \frac{5}{2}$ we are done.  We
know that $x$ received at least two credits by Lemma \ref{lem:even}, so
$cr(x)\ge \frac{3}{2}$, as the smallest credit is $\frac{3}{4}$.  If $cr(x) =
\frac{3}{2}$ then $x$ must have received two credits that were shared with
other dishonest variables.  Therefore, there are two clauses containing $x$
which are surely satisfied, and out of the other three the expected number of
unsatisfied clauses is $\frac{3}{2}\le cr(x)$.  Similarly, if $cr(x)=2$, then
$x$ shared a credit with another variable at least once, so one clause is
surely satisfied and the expected number of unsatisfied clauses out of the
other four is $2$.

We therefore have $\E[|U_2|] \le \sum_{x\in D} cr(x) \le \sum_{e\in E_G} w(e) +
\sum_{e\in E_S} ( w(e)-\frac{1}{2})$ and this concludes the proof. \qed

\end{proof}

\section{Conclusions}

We have given an alternative and (we believe) simpler inapproximability proof
for TSP, also modestly improving the known bound. We believe that the approach
followed here where the hardness proof goes explicitly through bounded
occurrence CSPs is more promising than the somewhat ad-hoc method of
\cite{pap06}, not only because it is easier to understand but also because we
stand to gain almost "automatically" from improvements in our understanding of
the inapproximability of bounded occurrence CSPs. In particular, though we used
the 5-regular amplifiers from \cite{BK03}, any such amplifier would work
essentially "out of the box", and any improved construction could imply an
improvement in our bound. Nevertheless, the distance between the upper and
lower bounds on the approximability of TSP remains quite large and it seems
that some major new idea will be needed to close it.

\bibliographystyle{abbrv} \bibliography{tsp}

\end{document}